\newcommand{\version}{April 23, 2020}
\documentclass[11pt,a4paper,reqno]{amsart}
\usepackage{amsthm,amsmath,amsfonts,amssymb,amsxtra,appendix,bookmark,dsfont,latexsym}

\makeatletter
\usepackage{hyperref}
\usepackage{delarray,color}
\usepackage{euscript}
\usepackage{bbm}

\makeatletter
\def\@tocline#1#2#3#4#5#6#7{\relax
  \ifnum #1>\c@tocdepth 
  \else
    \par \addpenalty\@secpenalty\addvspace{#2}%
    \begingroup \hyphenpenalty\@M
    \@ifempty{#4}{%
      \@tempdima\csname r@tocindent\number#1\endcsname\relax
    }{%
      \@tempdima#4\relax
    }%
    \parindent\z@ \leftskip#3\relax \advance\leftskip\@tempdima\relax
    \rightskip\@pnumwidth plus4em \parfillskip-\@pnumwidth
    #5\leavevmode\hskip-\@tempdima
      \ifcase #1
       \or\or \hskip 1em \or \hskip 2em \else \hskip 3em \fi%
      #6\nobreak\relax
      \dotfill
      \hbox to\@pnumwidth{\@tocpagenum{#7}}
    \par
    \nobreak
    \endgroup
  \fi}
\makeatother

\newcommand{\bdm}{\begin{displaymath}}
\newcommand{\edm}{\end{displaymath}}
\newcommand{\bdn}{\begin{eqnarray}}
\newcommand{\edn}{\end{eqnarray}}
\newcommand{\bay}{\begin{array}{c}}
\newcommand{\eay}{\end{array}}
\newcommand{\ben}{\begin{enumerate}}
\newcommand{\een}{\end{enumerate}}

\newcommand{\half}{\frac{1}{2}}

\newcommand{\N}{\mathbb{N}}
\newcommand{\R}{\mathbb{R}}

\newcommand{\eps}{\varepsilon}

\newcommand{\PsiLau}{\Psi_{\rm Lau}}
\newcommand{\cLau}{c _{\rm Lau}}

\newenvironment{rem}{\mbox{} \newline \noindent \textbf{Remark:}}{\vspace{0,3cm}}

\newtheorem{theorem}{Theorem}[section]
\newtheorem{lemma}[theorem]{Lemma}

\newtheorem{proposition}[theorem]{Proposition}

\theoremstyle{definition}

\theoremstyle{remark}

\newcommand{\beq}{\begin{equation}}
\newcommand{\eeq}{\end{equation}}

\numberwithin{equation}{section}

\newcommand{\as}{\mathrm{asym}}
\newcommand{\nLL}{n\mathrm{LL}}
\newcommand{\LLL}{\mathrm{LLL}}
\newcommand{\rv}{\mathbf{r}}

\newcommand{\qv}{\mathbf{q}}

\newcommand{\im}{\mathrm{i}}
\newcommand{\cL}{\mathcal{L}}

\begin{document}

\title[Holomorphic quantum Hall states in higher Landau levels]{Holomorphic quantum Hall states \\ in higher Landau levels}

\author{Nicolas Rougerie}
\address{Universit\'e Grenoble Alpes \& CNRS, LPMMC (UMR 5493), B.P. 166, F-38042 Grenoble, France}
\email{nicolas.rougerie@lpmmc.cnrs.fr}


\author{Jakob Yngvason}
\address{Faculty of Physics, University of Vienna, Boltzmanngasse 5, A-1090 Vienna, Austria}
\email{jakob.yngvason@univie.ac.at}



\date{\version}

\begin{abstract} 
Eigenstates of the planar magnetic Laplacian with homogeneous magnetic field form degenerate energy bands, the Landau levels. We discuss the unitary correspondence between states in higher Landau levels and those in the lowest Landau level, where wave functions are holomorphic. We apply this correspondence to many-body systems, in particular we represent  effective Hamiltonians and particle densities in higher Landau levels by corresponding quantities in the lowest Landau level.
\end{abstract}

\maketitle

\tableofcontents

\section{Introduction}

The state space of a charged particle moving in a homogeneous magnetic field in a plane orthogonal to the field decomposes into Landau levels, differing in energy by integral multiples of the magnetic field strength. When the position coordinates are expressed as complex numbers in the symmetric gauge the states in the lowest Landau level form a Bargmann space of holomorphic functions while wave functions in higher Landau  levels involve also powers
of the complex conjugate position variables in the standard representation. 

As noted by many people since long, and emphasized in  particular in \cite{Haldane-13,Haldane-18,CheBis-18}, a holomorphic representation of states  is not limited to the  lowest Landau level, where it has proved to be important for deriving some basic properties, e.g. \cite{LieRouYng-16,LieRouYng-17,OlgRou-19,Rougerie-xedp19,RouYng-17}. In fact, there is a natural unitary correspondence between states in different Landau levels, in particular between higher levels and the lowest one. 

In this expository paper we discuss several ways to arrive at the holomorphic representations  and derive explicit formulas for particle densities and effective Hamiltonians in higher Landau levels, expressed in terms of corresponding quantities in the lowest Landau level.  The methods have appeared in various disguises in the literature before but our aim is to  present them in a coherent fashion that, we hope, will be found useful for students and researchers in quantum Hall physics.

A physically appealing starting point is the decomposition of the position variables into guiding center variables and variables associated with the cyclotron motion of the particle around the guiding centers. While the components of the position operator commute, the other two sets of variables are non-commutative and satisfy canonical commutation relations. They can be represented in terms of creation and annihilation operators for two distinct and mutually commuting harmonic oscillators. One way of arriving at a holomorphic representation of states is an expansion in terms of coherent states for the harmonic oscillator of the guiding center variables~\cite{Haldane-18,CheBis-18}.  (These are the same as the \lq\lq vortex eigenstates\rq\rq in \cite{ChaFlo-07,ChaFloCan-08,ChaFlo-09}).  
The transformation between position coordinates and the coherent state variables can also  be expressed in terms of an integral operator with a kernel that is a modification of the reproducing kernel of a Bargmann space~\cite{ChaFlo-07}.

A formally simpler and more direct approach is to use the creation and annihilation operators of the cyclotron oscillator to define unitary mappings between different Landau levels\footnote{This approach appears already in~\cite{MacDonald-84} where it is attributed to Laughlin.}.
This gives explicit formulas for particle densities of many-body states in one Landau level in terms of polynomials in the Laplacian applied to corresponding densities in the lowest Landau level. The same formulas can alternatively be obtained by a Fourier transformation, exploiting the factorization of the exponential factor in the guiding center and cyclotron variables respectively. 
 
The main application of these considerations is in quantum Hall physics~\cite{Goerbig-09,Jain-07,Laughlin-99,StoTsuGos-99,Tong-16}. In this context, an electron gas is confined to two spatial dimensions and submitted to a magnetic field large enough to set the main energy scale. The quantization of the kinetic energy levels then becomes the salient feature. In the full plane, each level is infinitely degenerate, but for a finite area the degeneracy is proportional to the area times the field strength. For extremely large values of the latter, the lowest Landau level is degenerate enough to accommodate all electrons without violating the Pauli principle. 
For smaller values of the field several Landau levels can be completely filled with electrons and become inert in first approximation. The physics then boils down to the motion of the electrons in the last, partially filled, Landau level. 

In both cases only one Landau level has to be taken into account, and an effective model of widespread use in the literature is given in terms of a Hamiltonian acting on holomorphic functions. We review this first, before describing in more details the unitary mappings between Landau levels. The remarkable fact is that the dependence of the effective Hamiltonian on the Landau level it corresponds to is quite simple and transparent. An intuitive explanation (albeit not the most direct one from a  computational point of view) is that the  good variables to use are not the position variables but rather those of the guiding centers. The Landau level index, which fixes the energy of the cyclotron motion,  is encoded in a form factor in Fourier space that modifies external and interaction potentials via a differential operator.
In particular, the unitary mappings between Landau levels map multiplication by potentials to operators of the same kind.    
 
One salient feature of the effective operators acting on holomorphic functions is that they naturally suggest variational ans\"atze for their ground states, which become exact for certain truncated models. The Laughlin state~\cite{Laughlin-83,Laughlin-87} is the most emblematic of those, and much of our understanding of the fractional quantum Hall effect rests on its remarkable properties. In Sec.\ VI we apply our formulas to Laughlin states in an arbitrary Landau level, computing their density profiles and extending rigidity results from~\cite{LieRouYng-16,LieRouYng-17,RouYng-17,OlgRou-19}.

\bigskip

\section{Projected Hamiltonians and densities in quantum Hall physics}

Let us start from the many-body Hamiltonian (in symmetric gauge) for interacting 2D electrons in a constant perpendicular external magnetic field $B$ and a one-body potential~$V$
\begin{equation}\label{eq:full hamil}
H = \sum_{j=1} ^N \left( -\im \nabla_{\rv_j} + \frac{B}{2} \rv_j ^{\perp} \right) ^2 -NB+ V (\rv_j) + \sum_{i<j} w (\rv_i-\rv_j). 
\end{equation}
Here $w$ is the radial repulsive pair interaction potential, modeling 3D Coulomb interactions\footnote{Although electrons are confined to a 2D interface, they retain their interactions via the 3D Coulomb kernel.} in quantum Hall (QH) physics, but more general choices are also of interest. The one-body potential $V$ incorporates trapping in a finite size sample, plus the electrostatic potential generated by impurities. Mathematical conditions on the potentials $V$ and $w$ will be stated below. For convenience we have subtracted $NB$ from the kinetic part of the energy so that its lowest value is 0.
In the sequel vectors $\rv = (x,y) \in \R^2$ will very often be identified with complex numbers $z=x+{\mathrm i} y\in C$. 

As appropriate for electrons we consider the action of $H$ on the fermionic antisymmetric space
\begin{equation}\label{eq:full space}
 L^2_{\as} (\R^{2N}) = \bigotimes_{\as} ^N L^2 (\R^2). 
\end{equation}
For bosons one considers the symmetric tensor product instead; this is relevant for rotating cold atomic gases, where the rotation frequency takes over the role of the magnetic field.

In fractional quantum Hall (FQH) physics, the energy scales are set, by order of importance: first by the magnetic field, second by the repulsive interactions, third by the one-body potential. Our discussion in the sequel will reflect this. 



\subsection{Landau levels}
For large $B$ it is relevant to restrict particles to live in an eigenspace of $\left( -\im \nabla_{\rv} + \frac{B}{2} \rv ^{\perp} \right) ^2.$ Denote by
\begin{equation}\label{eq:nLL}
\nLL := \left\{ \psi \in L^2 (\R^2), \quad \left( -\im \nabla_{\rv} +\hbox{ $\frac{B}{2}$} \rv ^{\perp} \right) ^2 \psi =  2 B \left(n + \hbox{$\frac{1}{2}$} \right) \psi \right\}
\end{equation}
the $n$-th Landau level. The lowest level ($n=0$) will be denoted by $\LLL$; it is made of analytic $\times$ gaussian functions:
\begin{equation}\label{eq:LLL}
\LLL = \left\{ \psi (\rv) = f (x + \im y) e^{-\frac{B}{4} |\rv| ^2} \in L^2, \quad f \mbox{ analytic }  \right\}. 
\end{equation}
The corresponding fermionic spaces for $N$ particles will be denoted by $\nLL_N$ and $\LLL_N$:
\begin{equation}\label{eq:LLN}
\LLL_N = \bigotimes_\as ^N \LLL, \quad \nLL_N = \bigotimes_\as ^N \nLL.  
\end{equation}


\subsection{Hamiltonians in the LLL}
Consider projecting~\eqref{eq:full hamil} to the LLL. The first term is just a constant, the others can be expressed using the canonical basis 
\beq \varphi_m (z) = (\pi m!)^{-1/2} z^m e^{-\frac{B}{4} |z| ^2}.\eeq
Projecting~\eqref{eq:full hamil} to the LLL leads formally to 
\begin{equation}\label{eq:LLL hamil}
H_{w,V} ^{\LLL} = \sum_{j=1}^N \sum_{m,\ell \geq 0} \left\langle \varphi_m | V | \varphi_\ell \right\rangle \left| \varphi_m \right \rangle \left\langle \varphi_\ell \right |_j + \sum_{i<j} \sum_{m\geq 0} \left\langle \varphi_m | w | \varphi_m \right\rangle (|\varphi_m \rangle \langle \varphi_m |)_{ij} 
\end{equation}
where $(|\varphi_m \rangle \langle \varphi_m |)_{ij}$ projects\footnote{Note that fermionic wave-functions do not see the even $m$ terms of~\eqref{eq:LLL hamil}.} the relative coordinate $\rv_i - \rv_j$ on the state $\varphi_m$. Similarly $\left| \varphi_m \right\rangle \left\langle \varphi_\ell \right |_j$ is the operator mapping $\varphi_\ell$ to $\varphi_m$, acting on the $j$ variable only.

We assume that the potentials are measurable functions and that the \lq\lq moments\rq\rq 
\beq \langle \varphi_m|\, |V|\, |\varphi_m\rangle = \frac{1}{\pi m!}\int_{\mathbb R^2} |V(\mathbf r)|r^{2m} e^{-Br^2} d^2\mathbf r, \quad \langle \varphi_m|\, |w|\, |\varphi_m\rangle = \frac{1}{\pi m!} \int_{\mathbb R^2} |w(\mathbf r)| r^{2m}e^{-Br^2}d^2\mathbf r\eeq
are finite for all $m$. Then \eqref{eq:LLL hamil} is well defined as a quadratic form on a dense subspace of $\LLL_N$. Finiteness for $m=0$ means in particular that the potentials are in $L^1(\mathbb R^2)_{\rm loc}$ so derivatives of the potentials are well defined in the sense of distributions.  If the moments are uniformly bounded in $m$ and the potentials rotationnally symmetric (which implies the absence of terms $m\neq\ell$ in~\eqref{eq:LLL hamil}), then the corresponding operators are bounded and defined on the whole space. 

Usually in FQH physics one focuses attention on the interaction term in \eqref{eq:LLL hamil} (i.e., one sets $V\equiv 0$). There are no off-diagonal terms in it because $w$ is assumed to be radially symmetric. The coefficients $\left\langle \varphi_m | w | \varphi_m \right\rangle$ are often called ``Haldane pseudo-potentials'', cf. \cite{Haldane-1983}. If $w$ decreases rapidly at infinity then they also decrease rapidly with increasing $m$ and a basic observation in the theory of the fractional quantum Hall effect (FQHE) is that, if one truncates the sum~\eqref{eq:LLL hamil} at $m = \ell-1$, then the Laughlin state
\begin{equation}\label{eq:Laughlin}
\PsiLau ^{(\ell)} (z_1,\ldots,z_N) = \cLau ^{(\ell)} \prod_{i<j} (z_i-z_j) ^{\ell} e^{-\frac{B}{4} \sum_{j=1} ^N |z_j| ^2} 
\end{equation}
is an exact ground state ($L^2$-normalized by the constant in front). One can then argue, and prove to some extent~\cite{LieRouYng-16,LieRouYng-17,RouYng-17,OlgRou-19}, that such functions and natural variants are extremely robust, in particular to the addition of the external potential $V$.

\begin{rem} For very strong interaction potentials of range much smaller than the magnetic length $\sim B^{-1/2}$, in particular if there is a hard core, an expansion in terms of moments as in \eqref{eq:LLL hamil} is not adequate. This situation is analysed in  \cite{SY-20} which generalizes the paper \cite{LS-09}. It is shown that in an appropriate scaling limit the pseudo-potential operators $|\varphi_m\rangle\langle \varphi_m|$ also emerge, but with renormalized pre-factors involving the scattering lengths of the interaction potentials in the different angular momentum channels, rather than expectation values as in  \eqref{eq:LLL hamil}.\end{rem}


\subsection{Hamiltonians in higher Landau levels}
Consider now a situation where $n-1$ Landau levels are filled, so that additional electrons must sit in the higher ones, because of the Pauli principle. It is a common procedure in the FQH physics community~\cite{Jain-07,GoeLed-06,Tong-16} to model this situation using lowest Landau level (LLL) functions again. The basis for this reduction is the following statement, contained in one form or another  in a number of sources, in particular~\cite{Haldane-18,CheBis-18,Tong-16,Jain-07, MacDonaldGirvin-86, CifQui-10}. 

\begin{theorem}[\textbf{Effective Hamiltonian in the $n$-th Landau level}]\label{thm:main}\mbox{}\\
Let $H$ be given by~\eqref{eq:full hamil} and define
\beq H^{\nLL} = P^{\nLL} H P^{\nLL} \eeq
where $P^{\nLL}$ orthogonally projects all particles into the $\nLL$, i.e. it is the orthogonal projector from $L^2_{\as}(\R^{2N})$ to $\nLL_N$.

Then, for any $n$ there exists an effective external potential $V_n$ and an effective (radial) interaction potential $w_n$, depending only on $V,w$ and $n$ such that 
\beq H^{\nLL}- n\cdot 2BN \eeq
is unitarily equivalent to the LLL Hamiltonian $H_{V_n,w_n} ^{\LLL}$, defined as in~\eqref{eq:LLL hamil}, and acting on  $\LLL_N$. 

The effective $n$-th level potentials are as follows: 
\begin{align} V_n (\rv) &= L_n \left(-\mbox{$\frac 14$}{\Delta} \right) V (\rv) \label{eq:eff pot}
\\ 
 w_n (\rv) &= L_n \left(-\mbox{$\frac 14$} {\Delta}\right) ^2 w (\rv)\label{eq:eff int}
\end{align}
where $\Delta$ is the Laplacian and $L_n$ the Laguerre polynomial
\begin{equation}\label{eq:Laguerre pre}
L_n (u) = \sum_{l=0} ^n {n\choose l} \frac{(-u)^l}{l!}.  
\end{equation}
\end{theorem}
\begin{rem} Since we have not assumed any regularity of $V$ and $w$ except being measurable functions with finite moments the differentiations in \eqref{eq:eff pot} and \eqref{eq:eff int} have in general to be understood in the sense of distributions. This poses no problems, however, because the potentials are integrated against densities of wave functions in $\LLL_N$, which are smooth functions. Moreover, the densities have the form of polynomials times a gaussian so the finiteness of the moments for all $m$ guarantees that the integrals are well defined.
In Sec.\ V B it will be convenient to assume that the potentials have integrable Fourier transforms, but this is not really an extra restriction because the general case follows by a density argument. 
\end{rem}
\medskip

We shall give two proofs of the Theorem in Sec.~\ref{sec:proof thm}. Note that the constant we subtract from $H^{\nLL}$ is just the magnetic kinetic energy of $N$ particles in the $\nLL$.

\medskip

What the  Theorem says is that one can profit from the nice properties of the LLL to study phenomena in other Landau levels. This is particularly relevant because the main features are supposed not to depend very much on the potentials $V_n,w_n$ entering~\eqref{eq:LLL hamil}. In particular the Laughlin states 
have equivalents in any Landau level (cf Sec.~\ref{sec:Laughlin}). 
\bigskip

Since potential energies are integrals of potentials against particle densities, Theorem~\ref{thm:main} can be seen as a corollary of a general result about  particle densities of a many body states in different Landau levels. We recall that the $k$-particle density of an $N$-particle state with wave function $\Psi(\mathbf r_1,\dots, \mathbf r_N)$
is by definition
\beq 
\rho^{(k)}_\Psi(\mathbf r_1,\dots, \mathbf r_k)= {N \choose k} \int_{\mathbb R^{2(N-k)}} |\Psi(\mathbf r_1\cdots ;\mathbf r'_{k+1}\cdots \mathbf r'_N)|^2\mathrm d\mathbf r'_{k+1}\cdots \mathbf \mathrm dr'_N.
\eeq
If $\Psi\in \nLL_N$ for some $n$, then $\rho^{(k)}_\Psi$ is a $C^\infty$ function and decreases rapidly at infinity. This is discussed in Sec. V.

\begin{theorem}[\textbf{Particle densities in the $n$-th Landau level}]\label{thm:main2}\mbox{}\\
There is a unitary mapping $\mathcal U_{N,n}: \nLL_N\to \LLL_N$ such that if $\Psi_0=\mathcal U_{N,n}\Psi_n\in\LLL_N$ with $\Psi_n\in\nLL_N$ then for all $k$ 
\beq\label{eq:nLLdens}\rho^{(k)}_{\Psi_n}(\mathbf r_1,\dots, \mathbf r_k)
= \prod_{i=1}^k L_n \left(-\mbox{$\frac 14$}\Delta_{\mathbf r_i}\right)\rho^{(k)}_{\Psi_0}(\mathbf r_1,\dots, \mathbf r_k)
\eeq
\end{theorem}
Theorem~\ref{thm:main} follows as a corollary if one integrates  $V(\mathbf r)$ against the right hand side of \eqref{eq:nLLdens} with $k=1$, respectively $w(\mathbf r_1-\mathbf r_2)$ with $k=2$, and shifts the differentiations to the potentials by partial integration.

Conversely, Theorem \ref{thm:main2} (for $k=1,2$) follows from Theorem \ref{thm:main} if one regards the potentials as trial functions for the densities.
\medskip

In the following we shall define (in several related but distinct ways) the unitary mappings between Landau levels (see~\eqref{eq:unitary N}), and discuss the proofs of Theorems~\ref{thm:main} and ~\ref{thm:main2}. The physically most appealing way to interpret these unitaries is to see them as replacing the physical coordinates of electrons by the coordinates of the guiding centers of their cyclotron orbits, mathematically implemented through the use of coherent states. Indeed, in the LLL the position coordinates and the guiding center coordinates are really two different names for the same thing as will be evident in Sec.~\ref{sec:kernel}. Moreover, the quantum mechanical spread of both coordinates is of the order of the magnetic length $\sim B^{-1/2}$. The cyclotron radius in Landau level $n$ has an extra factor $\sqrt {n+1}$. Thus it is plausible that for large $B$ and small $n$ the difference between position and guiding center coordinates, and the non-commutativity of the latter, is not of much significance in thermodynamically large systems, i.e., for large $N$, provided the magnetic length stays much smaller than the interparticle distance.

Although the coherent state approach offers a satisfactory physical picture it is not always the most convenient one from a computational point of view. This motivates our review of alternate routes to the mappings between levels.

We also take the example of Laughlin states to explain how to deduce properties of the actual wave-functions in $\nLL_N$ minimizing effective energies from their representation in the $\LLL_N$ using the above unitary map. This amounts to saying that the density in guiding center coordinates can to a large extend indeed be identified with the true, physical, density in electron coordinates. We believe this is crucial for the understanding of the efficiency of the correspondence between Landau levels in FQH physics.  

\section{The Landau Hamiltonian and the two oscillators}

\subsection{The cyclotron oscillator}

The magnetic Hamiltonian of a particle of charge $q$ and effective mass $m^*$, moving in a plane with position variables  ${\mathbf r}=(x,y)$, is
\beq H=\frac 1{2m^*}(\pi_x^2+\pi_y^2)\eeq
where 
\beq \mbox{\boldmath$\pi$}=(\pi_x,\pi_y)=\mathbf p-q\mathbf A\eeq
is the gauge invariant kinetic momentum with $\mathbf A$ the magnetic vector potential and 
\beq\mathbf p=-\im \hbar (\partial_x,\partial_y)\eeq
the canonical momentum. We assume a homogeneous magnetic field of strength $B$ perpendicular to the plane and choose the
symmetric gauge  
\beq\mathbf A=\frac B2(-y,x).\eeq
Moreover, we choose units and signs so that $|q|=1$, $qB\equiv B>0$, $\hbar=1$ and $m^*=1$. Then
\beq \pi_x=-\im \partial_x+\half B y,\quad \pi_y=-\im \partial_y-\half B x\eeq
and the kinetic momentum components satisfy the canonical commutation relations (CCR)
\beq [\pi_x,\pi_y]=\im \ell_B^{-2}\label{CCRpi}\eeq
with 
\beq\ell_B=B^{-1/2}\eeq
the magnetic length.

In terms of the  creation and annihilation operators
\beq \quad a^\dagger=\frac {\ell_B}{\sqrt 2}(-\pi_y-\im \pi_x), \quad a=\frac{ \ell_B} {\sqrt 2 }(-\pi_y+\im \pi_x)\label{a}\eeq
with
 $[a,a^\dagger]=1$
the Hamiltonian is
\beq H=2B(a^\dagger a+\half). \eeq
Powers of $a^\dagger$ generate normalized eigenstates 
\beq\varphi_n=(n!)^{-1/2}(a^\dagger)^n\varphi_0\eeq 
with $a\varphi_0=0$ and the energy eigenvalues
\beq E_n=(n+\half)2B, n=1,2,\dots.\label{Landauspec}\eeq
In position variables the corresponding wave functions are
\beq \varphi_0({\bf r})=\frac 1{\sqrt {\pi}}\,e^{-(x^2+y^2)/4\ell_B}, \quad \hbox{and}\quad \varphi_n({\bf r})=\frac 1{\sqrt {\pi n!}}\, 
(x-\im y)^ne^{-(x^2+y^2)/4\ell_B}.\eeq

\subsection{Complex notation}\label{sec:complex}

With 
\begin{equation}
z =x+iy, \quad \bar z=x-\im y, \quad \partial_z=\half(\partial_x-\im \partial_y), \quad \partial_{\bar z}=\half(\partial_x+\im \partial_y) 
\end{equation}
we can write
\beq a^\dagger=\frac 1{\sqrt 2 \ell_B} (\half \bar z-2\ell_B^2 \partial_z),\quad   a=\frac 1{\sqrt 2 \ell_B} (\half z+2\ell_B^2 \partial_{\bar z}).\label{aellb}\eeq
Choosing units so that $B=2$, or equivalently, defining $z=\frac 1{\sqrt 2 {\ell_B}}(x+\im y)$, this becomes
\beq a^\dagger=\half \bar z- \partial_z, \quad  a= \half z+ \partial_{\bar z}.\eeq
Also, the gaussian factor $e^{-(|x|^2+|y|^2)/4\ell_B^2}$ becomes $e^{-|z|^2/2}$.

For computations it is often convenient to use the corresponding operators $\hat a^\dagger$, $\hat a$, acting on the pre-factors to the gaussian and defined by
\beq  a^{\#} \left[f(z,\bar z)e^{-|z|^2/2}\right]=\left[\hat a^{\#} f (z,\bar z)\right]e^{-|z|^2/2}.\eeq
These are
\beq \hat a^\dagger=\bar z-\partial_z, \quad  \hat a=\partial_{\bar z}.\label{hata}\eeq
In the sequel we shall generally use the hat $\hat{\phantom a}$ on operators and functions to indicate that the gaussian normalization factors are excluded.

Besides the standard definition $z=x+\im y$, other complexifications of $\mathbb R^2$ are  possible and can be useful, as stressed in~\cite{Haldane-18}. 

\subsection{The guiding center oscillator}\label{sec:guiding}

The classical 2D motion of a charged particle in a homogeneous magnetic field consists of a cyclotron rotation around \lq\lq guiding centers". The quantization of the cyclotron motion is the physical basis for the energy spectrum \eqref{Landauspec}, and the creation operators $a^\dagger$ generate the corresponding harmonic oscillator eigenstates. Every energy eigenvalue is infinitely degenerate, due to the different possible positions of the guiding centers. 

Quantum mechanically the dynamics of the guiding centers is described by another harmonic oscillator commuting with the first one. One arrives at this  picture by splitting the (gauge invariant) position operator $\bf r$ into the guiding center part $\bf R$ and the cyclotron part 
\beq \widetilde{\mathbf R}=\ell_B^2\mathbf n\times \mbox{\boldmath$\pi$},\eeq 
with $\mathbf n$  the unit normal vector to the plane. Both  $\bf R$  and $\widetilde{\bf R}$ are gauge invariant and they commute with each other. On the other hand the two components of $(R_x, R_y)$ of $\bf R$  do not commute  and likewise for the components of $\widetilde{\bf R}$. More precisely, we have
\beq \mathbf r=\mathbf R+\widetilde{\mathbf R}\label{splitting}\eeq
with
\beq R_x=x+\ell_B^2 \pi_y=\half x-\im \ell_B^2\partial_y, \quad R_y=y-\ell_B^2\pi_x=\half y+\im \ell_B^2\partial_x,\eeq
\beq \widetilde{R}_x=-\ell_B^2 \pi_y=\half x+\im \ell_B^2\partial_y, \quad \widetilde{R}_y=\ell_B^2\pi_x=\half y-\im \ell_B^2\partial_x\eeq
and the commutation relations
\beq[\mathbf R,\widetilde {\mathbf R}]=\mathbf 0,\quad
[R_x,R_y]=-\im \ell_B^2,\quad [\widetilde{R}_x,\widetilde{R}_y]= \im \ell_B^2.\label{CCR}
\eeq
The creation and annihilation operators for $\widetilde{\mathbf R}$ are the same as \eqref{a},
\beq a^\dagger=\frac1{\sqrt 2 \ell_B}(\widetilde{R}_x-\im \widetilde{R}_y),\quad  a=\frac1{\sqrt 2 \ell_B}(\widetilde{R}_x+ \im \widetilde{R}_y).\label{aa}\eeq
Those for the guiding center, on the other hand, are
\beq b^\dagger=\frac 1{\sqrt 2 \ell_B}(R_x+\im R_y),\quad b=\frac 1{\sqrt 2 \ell_B}(R_x-\im R_y).\label{b}\eeq
Note the different signs compared to \eqref{aa} due to the different signs in \eqref{CCR}. We have 
$[b,b^\dagger]=1$
and in complex notation
\beq   b^\dagger=\frac 1{\sqrt 2 \ell_B} (\half z-2\ell_B^2 \partial_{\bar z}),\quad b=\frac 1{\sqrt 2 \ell_B} (\half \bar z+2\ell_B^2 \partial_z).\label{bellb}\eeq
For $B=2$ 
\beq b^\dagger=\half z- \partial_{\bar z}, \quad  b= \half \bar z+ \partial_{z}\eeq
and 
\beq   \hat b^\dagger=z-\partial_{\bar z},\quad \hat b=\partial_{z}.\label{hatb}\eeq
The splitting \eqref{splitting} corresponds to
\beq
\left({\begin{array}{c} z\\ \bar z\\ \end{array}}\right)=\left({\begin{array}{c} b^\dagger\\ b\\ \end{array}}\right)+\left({\begin{array}{c} a\\ a^\dagger \\ \end{array}}\right)\label{splitting2}.
\eeq
While the operators $a^\dagger, a$ increases or decrease the Landau level index, the operators $b^\dagger, b$ leave each Landau level invariant. Pictorially speaking we can say that operators associated with the cyclotron oscillator move states \lq\lq vertically\rq\rq, i.e. act as ladder operators,  while those associated with the guiding center oscillator move them \lq\lq horizontally\rq\rq. 

With $\varphi_{0,0}=\varphi_0$ the common, normalized ground state for both oscillators, \beq a\varphi_{00}=b\varphi_{00}=0, \eeq 
the states 
\beq\varphi_{n,m}= \frac1{\sqrt{n!m!}}(a^\dagger)^n(b^\dagger)^m\varphi_{0,0}=\frac1{\sqrt{n!m!}}(b^\dagger)^m(a^\dagger)^n\varphi_{0,0},\quad n,m=0,1,\dots\eeq
form a basis of common eigenstates of the oscillators with $\varphi_{n,0}$ being the previously  defined $\varphi_n$. For fixed $n$ the states $\varphi_{n,m}$, $m=0,1,\dots$ generate the Hilbert space of the $n$'th Landau level, which we shall denote by $n$LL. The lowest Landau level will be denoted LLL.

Using complex coordinates the wave functions with $n=0$ respectively $m=0$ are
\beq \varphi_{0,m} (z,\bar z)=\frac 1{\sqrt{\pi m!}}\, z^m e^{-|z|^2/2},\quad \varphi_{n,0}(z,\bar z)=\frac 1{\sqrt {\pi n!}}\, \bar z^ne^{-|z|^2/2}.\eeq
More generally, the wave functions
\beq \varphi_{n,m}(z,\bar z)=\frac 1{\sqrt {\pi n!\,m!}} [(z-\partial_{\bar z})^m \bar z^n ] e^{-|z|^2/2}=\frac 1{\sqrt {\pi n!\,m!}}[(\bar z-\partial_{z})^n z^m ] e^{-|z|^2/2}\label{nLLbasis}\eeq
can be written in terms of associated Laguerre polynomials. They are eigenfunctions of the angular momentum operator in the symmetric gauge (acting on the pre-factor to the gaussian)
\beq \hat L=z\partial_z-\bar z\partial_{\bar z}\eeq
with eigenvalues $M=-n+m=-n,-n+1,\dots$ in the $n$LL. The operators $b^\dagger, b$ shift the angular momentum within each Landau level. 

\section{Expressions of the inter-level unitary maps}

\subsection{With coherent states}

A coherent state associated with the guiding center oscillator in the $n$LL with parameter $Z\in\mathbb C$ is defined in a standard way~\cite{ComRob-12,KlaSka-85} as
\beq |Z,n\rangle=e^{(Zb^\dagger -\bar Zb)}\varphi_{n,0}=\sum_{m=0}^\infty \frac {Z^m}{\sqrt {m!}}\, \varphi_{n,m} e^{-|Z|^2/2}\label{cohstate}.\eeq
The overlap of two coherent states is
\beq
\langle Z,n|Z',n'\rangle=\delta_{n,n'}  e^{(2\bar Z Z'-|Z|^2-|Z'|^2)/2}=\delta_{n,n'}  e^{-|Z-Z'|^2/2}\,e^{{\mathrm i }\,\text{Im}\,(\bar Z Z')}.\eeq
Moreover, 
\beq
\int |Z,n\rangle\langle Z,n|\,\frac{{\mathrm d}^2Z}\pi =\Pi_n\label{nproj}
\eeq
is the projector on $n$LL, where ${\mathrm d}^2Z:=\hbox{$\frac{ \mathrm i }{2}$} {\mathrm d}Z\wedge {\mathrm d}\bar Z$ is the Lebesgue measure on the plane. Indeed, 
\beq \frac {1}{\sqrt { m!\, m'!}} 
\int \bar Z^m Z^{m'} e^{-|Z|^2} \frac{{\mathrm d}^2Z}\pi=\langle \varphi_{n,m}|\varphi_{n,m'}\rangle=\delta_{m,m'},\eeq
and
\beq
\Pi_n=\sum_{m=0}^\infty |\varphi_{n,m}\rangle\langle \varphi_{n,m}|.
\eeq

The coherent states allow an interpretation of $n$LL as a Bargmann space of analytic functions of the coherent state variable $Z$: If $\psi\in$\,$n$LL then
\beq \widehat{\Psi}(Z):= \langle \bar Z,n|\psi\rangle e^{|Z|^2/2}=\sum_{m=0}^\infty\langle \varphi_{n,m}|\psi\rangle\frac {Z^m}{\sqrt {m!}}\, \label{analyt}\eeq
is analytic in $Z$ and 
\beq \Psi(Z,\bar Z)=\widehat{\Psi}(Z)e^{-|Z|^2/2}\label{analyt2}\eeq
has the same $L^2$ norm as $\psi$ because of \eqref{nproj}. Thus the map 
\beq U_n:\psi\mapsto \Psi\eeq
is isometric from the $n$LL to the  LLL. From the definition it is clear that 
\beq U_n\varphi_{n,m}=\varphi_{0,m}\label{Un1}\eeq 
and
\beq U_n|Z,n\rangle=|Z,0\rangle,\label{Un2}\eeq
so  $U_n$ is in fact a unitary with 
\beq U_n^{-1}\varphi_{0,m}=\varphi_{n,m}\quad \hbox{and}\quad
U_n^{-1}|Z,0\rangle=|Z,n\rangle.\label{Un-1}\eeq
Either \eqref{Un1} or \eqref{Un2} can be taken as the definition of $U_n$. The unitary map 
\begin{align}\label{eq:unitary N}
\mathcal{U}_{N,n}: \nLL_N  &\rightarrow \LLL_N \nonumber\\
\Psi_N &\mapsto \left(\bigotimes_\as ^N U_n \right) \Psi_N
\end{align}
is that used in Theorem~\ref{thm:main}.

The function $\Psi(Z,\bar Z)$ coincides with the LLL wave function of $U_n\psi$ if $Z$ is identified with the complex position variable $z=x+\im y$. Note, however, 
that $Z$ is associated with the (non-commutative) components of the guiding center operator $\mathbf R$ rather than the (commutative)  position operator $\mathbf r$. By the definition \eqref{analyt} $\Psi$ depends linearly on $\psi$; the alternative definition $\Psi=\langle\psi|Z,n\rangle$, that is sometimes used, leads to an anti-unitary correspondence.

\subsection{With integral kernels}\label{sec:kernel}

Consider the coherent state \eqref{cohstate} without the gaussian normalization factors as as function of $Z; z,\bar z$: 
\beq \widehat{F}(Z; z,\bar z)=\sum_{m=0}^\infty \frac {Z^m}{\sqrt {m!}}\, \widehat{\varphi}_{n,m}(z,\bar z).\eeq
The coherent state is an eigenstate of the annihilation operator $\hat b=\partial_z$ with eigenvalue $Z$, so
\beq \widehat{F}(Z; z,\bar z)=f(z,\bar z) e^{Zz} .\eeq
Furthermore, $\hat F$ is an eigenstate of $\hat a^\dagger\hat a=(\bar z-\partial_z)\partial_{\bar z}$ to eigenvalue $n$ which leads to
\beq \widehat{F}(Z; z,\bar z)=c_n(\bar z-Z)^n e^{Zz}\eeq
with a normalization constant $c_n=1/\sqrt{\pi n!}$. The full coherent state \eqref{cohstate} as a function of $Z$, $z$ and $\bar z$, including normalization factors,  is thus given by
\beq c_n(\bar z-Z)^n e^{-(|Z|^2+|z|^2-2Zz)/2}.\eeq
Inserting this into \eqref{analyt}  gives
\beq \Psi(Z,\bar Z)=\int G(Z,\bar Z;z,\bar z)\psi(z,\bar z)\,{\mathrm d}^2z \label{psiPsi}\eeq
with
\beq G(Z,\bar Z;z,\bar z)= \frac 1{ \sqrt {\pi n!}}\,(z-Z)^n e^{-(|Z|^2+|z|^2-2Z\bar z)/2}= \frac 1{ \sqrt {\pi n!}}\,(z-Z)^n e^{-|z-Z|^2/2}
e^{-{\mathrm i }\,\text{Im}\, (\bar z Z)}.\label{G}
\eeq
This formula was derived in a different way in \cite{ChaFlo-07} and appears there (in slightly different notation) as Equation~(34). The inverse map is given by 
\beq \psi(z,\bar z)=\int \bar G(z,\bar z; Z,\bar Z)\Psi(Z,\bar Z)\,{\mathrm d}^2Z\label{Psipsi}
\eeq
with
\beq
\bar G(z,\bar z; Z,\bar Z)=\frac 1{ \sqrt {\pi n!}}\,(\bar z-\bar Z)^n e^{-(|Z|^2+|z|^2-2\bar Z z)/2}= \frac 1{ \sqrt {\pi n!}}\,(\bar z-\bar Z)^n e^{-|z-Z|^2/2}
e^{-{\mathrm i }\,\text{Im}\, (z\bar Z)}.\label{Gbar}\eeq
Note that $G$ can be written as
\beq g(z-Z)\,e^{-2{\mathrm i }\,\text{Im}\, (\bar z Z)}\eeq
where $g$ is essentially concentrated  in a disc of radius $\sim \sqrt{ n+1}$ and the factor is a phase factor. Recall also that the length unit is $\sqrt 2\ell_B\sim B^{-1/2}$.

A further remark is that for $n=0$ $G$ is the reproducing kernel in Bargmann space, confirming again that in the LLL $\Psi$ and $\psi$ are the same function on $\mathbb C$ just with different names for the variables. The phase factor in $G$ is essential for this to hold.

\subsection{With ladder operators}

A direct approach to the correspondence $n$LL $\leftrightarrow$ LLL, by-passing the coherent states, starts from \eqref{Un1}, noting that 
\beq U_n= (n!)^{-1/2} a^n\quad\hbox{restricted to $n$LL}\eeq
and hence
\beq U_n^{-1}= (n!)^{-1/2} (a^\dagger)^n\quad\hbox{restricted to LLL}.\eeq
Using the representations \eqref{hata} for the creation and annihilation operators we conclude that the following holds:  

\begin{proposition}[\textbf{Unitary maps with ladder operators}]\label{lem:simple}\mbox{}\\
Let $\psi_n\in\nLL$ have wave function
\beq\label{eq:nLLfunc}
\psi_n(z,\bar z)=\sum_{\nu=0}^n \bar z^\nu f_\nu(z)e^{-|z|^2/2},\eeq 
$f_\nu$ analytic for $\nu=0,\dots, n$. Then $\Psi_0 = U_n \psi_n\in\LLL$ has wave-function
\beq
\Psi_0(z,\bar z)= \sqrt{n!}f_n(z)e^{-|z|^2/2}.
\eeq
Conversely, the wave function of $\psi_n = U_n^{-1}\Psi_0$ is
\begin{align}\label{mainformula} 
\psi_n(z,\bar z) &=  [(\bar z-\partial_z)^nf_n(z)]e^{-|z|^2/2}\nonumber\\
&= \left[\bar z^n f_n(z)+\sum_{k=1}^n(-1)^k {n\choose k} \,\bar z^{n-k}f_n^{(k)}(z)\right]e^{-|z|^2/2}.
\end{align}
\end{proposition}

Note that Equation~\eqref{mainformula} implies in particular that the factor $f_n(z)$ to the highest power $n$ of $\bar z$ determines uniquely the factors to the lower powers $\bar z^\nu$:
\beq f_\nu(z)= (-1)^{n-\nu}\left({\begin{array}{c} n\\ \nu\\ \end{array}}\right) f_n^{(n-\nu)}(z).\label{134}\eeq
The state is thus completely fixed by the holomorphic function $f_n$ and the Landau index $n$.
\bigskip

Incidentally, these considerations also lead to a method for projecting functions to the lowest Landau level:

\begin{proposition}[\textbf{LLL projection}]\mbox{}\\
Let
\beq \phi(z,\bar z)=\sum_{\nu=0}^n \bar z^\nu g_\nu(z)e^{-|z|^2/2}\label{135}\eeq
with arbitrary analytic functions $g_\nu$.  It's orthogonal projection into $\LLL$ is 
\beq\label{eq:LLL projection}
\mathcal P_{\rm LLL}\phi(z)=\sum_{\nu=0}^n g^{(\nu)}_\nu(z) e^{-|z|^2/2} .
\eeq
where $g^{(\nu)} = \partial_z ^{\nu} g$. 
\end{proposition}

This is well-known as the recipe ``move all $\bar{z}$ factors to the left and replace them by derivatives in $z$'', see e.g.~\cite{Jain-07}. For completeness we give the simple proof:

\begin{proof}
The previous considerations lead to a method for splitting a state 
\beq\phi\in\bigoplus_{k=0}^n \: k\mathrm{LL}\eeq
into its components in the different LL: Start with a wave function as in~\eqref{135}. It's component $\psi_n$ in the $n$LL is then given by \eqref{mainformula} with $f_n:=g_n$ and $f_\nu$ for $0\leq\nu\leq n-1$ defined by~\eqref{134}. The difference
$\tilde \phi=\phi-\psi_n$ is now in $\bigoplus_{k=0}^{n-1} k\mathrm{LL}$ and we can repeat the procedure with $n$ replaced by $n-1$, $\phi$ by $\tilde\phi$  etc. until we obtain the splitting
$\varphi=\sum_{k=0}^n \psi_{k}$
with $\psi_{k}\in k\mathrm{LL}$.
By induction over $n$, using that 
\beq 
\sum_{\nu=0}^n {n\choose \nu} (-1)^\nu=(1-1)^n=0\eeq 
this procedure implies~\eqref{eq:LLL projection}.
\end{proof}

\subsection{Recap of the different expressions for the unitary maps}

Summarizing the contents of this section, we have displayed three equivalent ways to represent a state $\Psi\in n$LL by analytic functions in Bargmann space:
\begin{itemize}
\item[{\bf 1.}] Take the scalar product 
 $\langle \bar Z,n|\Psi\rangle $
with a coherent state, cf. \eqref{analyt}.
\item[{\bf 2.}] Use Equation \eqref{psiPsi} with the integral kernel \eqref{G}.
\item[{\bf 3.}]  Apply the differential operator $\partial_{\bar z}$ $n$-times to the pre-factor of the Gaussian. Equivalently: Expand the pre-factor in powers of $\bar z$ and keep only the highest power. The inverse mapping, LLL $\to$ $n$LL, is achieved by applying the differential operator $(\bar z-\partial_z)^n$ to the analytic function representing the state in the LLL.
\end{itemize}

The last method is formally the simplest and in the Sec. V we shall use it to discuss  particle densities in higher Landau levels in term of their counterparts in the lowest Landau level. 

\section{Particle densities and the $\nLL$ Hamiltonian, proofs of the Theorems}\label{sec:proof thm}

We now have all the necessary ingredients to prove Theorems~\ref{thm:main} and~\ref{thm:main2}. We provide two slightly different approaches.

\subsection{Many body states and particle densities}\label{densities}

All considerations in Secs. III and IV carry straightforwardly over to many-body states in symmetric or anti-symmetric tensor powers $\nLL_N\equiv
 n$LL$^{\otimes_{{\rm s,a}}N}$ of single particle states by applying the single particle formulas to each tensor factor. 
 
Let $\Psi_n$ be a state in $\nLL_N$ with wave function 
\beq
\psi_n(z_1,\bar z_1; \dots ;z_N,\bar z_N)=\widehat{\psi}_n(z_1,\bar z_1; \dots ;z_N,\bar z_N)e^{-(|z_1|^2+\cdots +|z_N|^2)/2}.
\eeq
Expanding in powers of $\bar z_i$ we can write
\beq 
\widehat{\psi}_n(z_1,\bar z_1; \dots ;z_N,\bar z_N)=\prod_{i=1}^N \bar z_i^n f_n(z_1,\dots, z_N)+\sum\prod_{i=1}^N\bar z_i^{\nu_i} f_{\nu_1,\dots, \nu_N} (z_1,\dots, z_N).
\eeq
The sum is here over $N$-tuples $(\nu_1,\dots\nu_N)$ such that $\nu_k<n$ for at least one $k$. The functions $f_n$ and $f_{\nu_1,\dots, \nu_N}$ are holomorphic and the latter are, in fact, derivatives of $f_n$, cf. \eqref{134}.

The state $\Psi_0=U_n\Psi_n$ in LLL$^N$  has the wave function
\beq \psi_0(z_1,\bar z_1; \dots ;z_N,\bar z_N)=\widehat \psi_0(z_1,\dots, z_N)e^{-(|z_1|^2\cdots |z_N|^2)/2}\eeq
with 
\beq \widehat \psi_0(z_1,\dots, z_N)=(n!)^{-N/2}\prod_{i=1}^N \partial_{\bar z_i}^n \, \widehat \psi_n(z_1,\bar z_1; \dots ;z_N,\bar z_N)=
(n!)^{N/2}f_n(z_1,\dots, z_N).
\eeq 
The wave function $\widehat\psi_n$ can now be written
\beq \widehat\psi_n((z_1,\bar z_1; \dots ;z_N,\bar z_N)=(n!)^{-N/2}\prod_{i=1}^N (\bar z_i-\partial_{z_i})^n\widehat \psi_0(z_1,\dots, z_N)=
\prod_{i=1}^N (\bar z_i-\partial_{z_i})^nf_n(z_1,\dots, z_N).\eeq
Next we consider the  $k$-particle density of $\Psi_n$, defined by 
 \begin{multline} \rho^{(k)}_n(z_1,{\bar z}_1, \dots, z_k,\bar z_k)= {N \choose k} \int |\psi_n(z_1,\bar z_1; \dots ;z_N,\bar z_N)|^2 {\mathrm d}^2z_{k+1}\cdots {\mathrm d}^2z_N\\={N \choose k} \int \left|\widehat \psi_n(z_1,\bar z_1; \dots ;z_N,\bar z_N)\right|^2 e^{-(|z_1|^2+\cdots +|z_N|^2)} {\mathrm d}^2z_{k+1}\cdots {\mathrm d}^2z_N.\label{densmatr}\end{multline}
The density $\rho^{(k)}_0$ of $\Psi_0=U_n\Psi_n$ is given by the same formula with $n=0$. 

Functions in LLL are holomorphic and decrease at $\infty$ as $e^{-|z|^2/2}$; the latter follows from the fact that the Bargmann kernel \eqref{G} with $n=0$, which has this decrease, is a reproducing kernel for the Hilbert space LLL. Equation~\eqref{G} (equivalently, Equation \eqref{eq:nLLfunc}) also implies that functions in $\nLL$ are $C^\infty$ in the real position variables and decrease in the same way. This clearly carries over to wave functions in $\nLL_N$ and corresponding densities.
 
To prove Theorem~\ref{thm:main2} (which then implies Theorem~\ref{thm:main}) we have to compare   $\rho^{(k)}_n$ and $\rho^{(k)}_0$. It is, in fact, sufficient to consider the problem for a single variable, i.e., to prove the following Lemma:
\begin{lemma}[\textbf{Reshuffling differentiations}]\label{lem:proj dens}\mbox{}\\
Let $\psi(z,\bar z)=f(z)e^{-|z|^2/2}$ with holomorphic $f$. 
Then
\begin{equation}
(n!)^{-1}\overline {\left[(\bar z-\partial_{z})^n f(z)\right]}\left[(\bar z-\partial_{z})^n f(z)\right]e^{-z \bar z}= 
L_n\left(- \partial_{\bar z}\partial_z\right)\left[\overline{f(z)}f(z)e^{-z \bar z}\right]\label{nlift0} 
\end{equation}
with $L_n$ the Laguerre polynomial~\eqref{eq:Laguerre pre}. Recall that $\partial_{\bar z} \partial_z=\frac 14 \Delta$. 
\end{lemma}

\begin{proof} This is a straightforward computation by induction over $n$, using the recursion relation for the Laguerre polynomials,
\beq (n+1)L_{n+1}(u)=(2n+1)L_n(u)-nL_{n-1}(u)-uL_n(u).\eeq
To compute \beq\partial_{\bar z}\partial_z\left[ \overline {\left[(\bar z-\partial_{z})^n f(z)\right]}\left[(\bar z-\partial_{z})^n f(z)\right]e^{-z \bar z}\right],\eeq
starting with $n=0$ and $L_0=1$, one uses the commutation relations
\beq
\partial_{\bar z} (\bar z-\partial_z)^n=(\bar z-\partial)^n\partial_{\bar z}+n(\bar z-\partial_z)^{n-1},\quad \partial_{z} (\bar z-\partial_z)^n=(\bar z-\partial)^n\partial_{z}, 
\eeq
and the fact that $\partial_{\bar z}f(z)=\partial_{z}\overline{f(z)}=0$ for holomorphic $f$. 
\end{proof}

Applying the Lemma to each variable in a many-body wavefunction leads directly to \eqref{eq:nLLdens} and hence Eqs.\eqref{eq:eff pot} and \eqref{eq:eff int}.

\subsection{Projected Hamiltonian and guiding center coordinates}

We now discuss an alternative road to~\eqref{eq:eff pot} and~\eqref{eq:eff int}, providing additional insights. The starting point  is the splitting \eqref{splitting} of the position variables in guiding centers and cyclotron motion, and the ensuing factorization of matrix elements of  $\exp({\mathrm i}\mathbf q\cdot \mathbf r) $ which enter the Fourier transformed version of \eqref{nlift0}.

\begin{lemma}[\textbf{Plane waves projected in Landau levels}]\label{lem:plane}\mbox{}\\
For any $\mathbf{q}\in \R ^2$, identify $e^{\im \mathbf{q}\cdot \rv}$ with the corresponding multiplication operator on $L^2 (\R^2)$, where $\rv$ is the spatial variable. Let $\mathbf{R}$ be the guiding center operator defined in Sec.~\ref{sec:guiding}, $\Pi_n$ the orthogonal projector on $\nLL$ and $U_n:\nLL \rightarrow \LLL$ the inter-LL unitary map.

We have that 
\begin{equation}\label{eq:proj wave}
U_n \Pi_n e^{\im \mathbf{q}\cdot \rv} \Pi_n U_n ^* = L_n \left(\frac{|\mathbf{q}|^2}{4}\right) e^{-\frac{|\mathbf{q}|^2}{8}} \Pi_0  e^{\im \mathbf{q}\cdot \mathbf{R}} \Pi_0 = L_n \left(\frac{|\mathbf{q}|^2}{4}\right)  \Pi_0  e^{\im \mathbf{q}\cdot \mathbf{r}} \Pi_0 
\end{equation}
with the Laguerre polynomial
\begin{equation}\label{eq:Laguerre}
L_n (u) = \sum_{l=0} ^n {n\choose l} \frac{(-u)^l}{l!}.  
\end{equation}
\end{lemma}

The equality of the left-hand and the right-hand sides of~\eqref{eq:proj wave} can be seen as a Fourier transformed version of~\eqref{eq:nLLdens} (with $k=1$). The identity \eqref{eq:proj wave} implies that the norm of $\Pi_0 e^{\im \mathbf{q}\cdot \mathbf{r}} \Pi_0$ decays faster than any polynomial in $|q|$. Indeed, on the left hand side we have a product of unitaries and projections whose norm is bounded by one. Also, when $q\neq 0$ and $n$ grows, the norm of $\Pi_n e^{\im \mathbf{q}\cdot \mathbf{r}} \Pi_n$ decays like $n^{-1/4}$.

We now provide another proof, using guiding center coordinates rather than ladder operators. This also connects with the middle expression in~\eqref{eq:proj wave}.

\begin{proof}
Some of the following computations can be found in a variety of sources, e.g.~\cite[Proof of Theorem~3.2]{Jain-07} or~\cite{GoeLed-06}.

First note that if $A$ is a function of $a^\dagger, a$ and $B$ of $b^\dagger, b$, then
\beq
\langle \varphi_{n',m'}|AB|\varphi_{n,m}\rangle=\langle \varphi_{n',0}|A|\varphi_{n,0}\rangle\, \langle \varphi_{0,m'}|B|\varphi_{0,m}\rangle.\label{factorization}
\eeq
This is a consequence of the fact that the two commuting harmonic oscillators \eqref{a} and \eqref{b} can be represented, in a unitarily equivalent way,  in the tensor product of two spaces with basis vectors $\varphi_{n,0}$ and $\varphi_{0,m}$ respectively. In this representation $\varphi_{n,m}=\varphi_{n,0}\otimes \varphi_{0,m}$ and the operators $A$ and $B$  act independently on each of the tensor factors. One can also pick directly $A,B$ to be polynomials in creation and annihilation operators and use the CCR to prove the claim.
Note, however, that in the representation \eqref{nLLbasis} the functions $\varphi_{n,m}(z,\bar z)$ are not simply products of the functions $\varphi_{n,0}$ and $\varphi_{0,m}$. Indeed, the  variables $z$ and $\bar z$, regarded as position operators, do not act independently in the tensor factors, cf. \eqref{splitting2}.

We apply \eqref{factorization} to compute the matrix elements of $\exp({\mathrm i}\mathbf q\cdot \mathbf r) $, $\mathbf q\in\mathbb R^2$, in the nLL. With 
\beq\mathbf q=(q_x,q_y), \quad q=q_x+\im q_y,  \quad \mathbf r=(x,y), \quad z=x+iy\eeq
and further employing
\beq z=a+b^\dagger, \qquad \bar z=a^\dagger+b\eeq
we have
\beq 
\mathbf q\cdot \mathbf r=q_x x+q_y y=\half(\bar q z+q\bar z)=\half (q a^\dagger+\bar q a)+\half(\bar q b^\dagger+q b). \label{bfr}
\eeq
Since the $a^\#$'s and the $b^\#$'s commute it follows that
\beq
\exp\left({\mathrm i}\mathbf q\cdot \mathbf r\right)=\exp \left(\frac {\mathrm i} 2 (q a^\dagger+\bar q a) \right)
\exp \left(\frac {\mathrm i} 2(\bar q b^\dagger+q b)\right)
\eeq
and thus by \eqref{factorization}
\beq
\langle \varphi_{n,m'}|\exp({\mathrm i}\mathbf q\cdot \mathbf r)|\varphi_{n,m}\rangle=
\langle \varphi_{n,0}|\exp \left(\frac {\mathrm i} 2 (q a^\dagger+\bar q a) \right)
|\varphi_{n,0}\rangle\langle \varphi_{0,m'}| \exp \left(\frac {\mathrm i} 2(\bar q b^\dagger+q b)\right)|\varphi_{0,m}\rangle.
\eeq
By the Baker-Campbell-Hausdorff formula 
\beq e^{X+Y}=e^X\,e^Y\,e^{-\half[X,Y]}\eeq
for two operators commutating with their commutator (recall that $[a, a^\dagger]=1$) we can write
\beq
\exp \left(\frac {\mathrm i} 2 (q a^\dagger+\bar q a) \right) = \exp \left( \frac {\mathrm i} 2 (q a^\dagger) \right) \exp\left( \frac {\mathrm i} 2 (\bar q a)\right) \exp\left(-\frac 1 8|q|^2\right)
\eeq
and thus
\begin{equation}
\langle \varphi_{n,m'}|\exp\left({\mathrm i}\mathbf q\cdot \mathbf r\right)|\varphi_{n,m}\rangle= \tilde h_n (q)
\langle \varphi_{0,m'}| \exp \left(\frac {\mathrm i} 2(\bar q b^\dagger+q b)\right)|\varphi_{0,m}\rangle.\label{expvalues}
\end{equation}
with 
\beq
\tilde h_n(q)=\langle \exp\left(\frac {-\mathrm i} 2 (\bar q a) \right)\varphi_{n,0}| \exp \left(\frac {\mathrm i} 2 (\bar q a)\right)\varphi_{n,0}\rangle \exp\left(-\frac 1 8 |q|^2\right).
\eeq
Expanding the exponential and using $a^k\,\varphi_{n,0}=\sqrt {(n-1)\cdots(n-k+1)}\, \varphi_{n-k,0}$ we obtain
\begin{equation}
\tilde h_n(q)=\sum_{k=0}^n \frac {(-1)^k}{4^{k}}{n\choose k}\frac {1}{k!} |q|^{2k} \exp\left(- \frac 1 8 |q|^2\right)=L_n(\mbox{$\frac 14$}|q|^2) \exp\left(- \frac 1 8 |q|^2\right).\label{8.9} 
\end{equation}
Thus, recalling the definition of the guiding center coordinate $\mathbf{R}$ in Sec.~\ref{sec:guiding},~\eqref{expvalues} implies the first equality in~\eqref{eq:proj wave}.

To obtain the second equality we  subtract
\beq\frac {\mathrm i} 2 (q a^\dagger+\bar q a)\eeq
from \eqref{bfr} to get, employing the Campbell-Hausdorff formula again,
\begin{multline}
 \exp\left({\mathrm i}\mathbf q\cdot \mathbf R\right)
 =\exp\left(-\frac {\mathrm i} 2 q a^\dagger+ \left(
 \frac {\mathrm i} 2 (q a^\dagger+\bar q a)
+\frac {\mathrm i} 2 (\bar q b^\dagger+q b)\right)- \frac {\mathrm i} 2 \bar q a\right)\\
=\exp\left(-\frac {\mathrm i} 2 q a^\dagger\right)\exp\left(\mathrm i\mathbf q\cdot\mathrm r\right)\exp\left(-\frac {\mathrm i} 2 \bar q a\right)\exp\left(\frac 1 8|q|^2\right).
\end{multline}
On the LLL $\exp\left(-\frac {\mathrm i} 2 \bar q a\right)$ is the identity,  so the second equality in~\eqref{eq:proj wave}  follows.
\end{proof}

To deduce Equations~\eqref{eq:eff int}--\eqref{eq:eff pot} and hence Theorem 2.1 it only remains to write the Fourier decompositions 
\beq V (\rv) = \int_{\R^2} \widehat{V} (\qv) e^{\im \qv \cdot \rv} d\qv\eeq
and 
\beq w(\rv_1 - \rv_2) = \int_{\R^2} \widehat{w} (\qv) e^{\im \qv \cdot \rv_1} e^{-\im \qv \cdot \rv_2} d\qv \eeq
and use Lemma~\ref{lem:plane}. The expressions involving Laplacians in Theorem~\ref{thm:main} follow from the Fourier representation $-\Delta = |\qv|^2.$ This argument demands absolute integrability of the Fourier transforms, but the general case follows by a density argument.

\section{Laughlin states in higher Landau levels}\label{sec:Laughlin}

As already mentioned, a crucial approximation in FQH physics is to truncate the Haldane 
pseudo-potential series in the LLL Hamiltonian~\eqref{eq:LLL hamil} to obtain the Laughlin state~\eqref{eq:Laughlin} as an exact ground state of the translation invariant problem $V\equiv 0$. 

In view of Theorem~\ref{thm:main}, it is desirable to do the same in a higher Landau level, at the level of the effective Hamiltonian~$H_{0,w_n} ^{\LLL}$ and~\eqref{eq:Laughlin} then becomes an exact ground state after the unitary mapping to the LLL. In this section, we explain how the previous considerations allow to study the properties of the corresponding physical wave-function (that is, as expressed in the position coordinates, rather than in the guiding center coordinates). 

\subsection{Density estimates on mesoscopic scales}

Consider a Laughlin state in the $\LLL_N$
\beq
\Psi_{0,N}^{\rm Lau}=c_N\prod_{i<j}(b^\dagger_i-b^\dagger_j)^\ell \varphi_{0,0}^{\otimes N}\label{Laugh01}
\eeq
with wave function
\beq
\Psi_{0,N}^{\rm Lau}(z_1,\dots,z_N)=c_N\prod_{i<j}(z_i-z_j)^\ell e^{-(|z_1|^2+\cdots+|z_N|^2)/2}.\label{Laugh02}
\eeq
Here $\ell=1, 3,\dots$ for fermions and $\ell=2,4,\dots$ for bosons. We denote by  
\begin{equation}\label{eq:Laugh dens}
\varrho^{\rm Lau}_{0,N}(\rv):= N \int_{\R ^{2(N-1)}} |\Psi_{0,N} (\rv,\rv_2,\ldots,\rv_N)|^2 d\rv_2 \ldots d\rv_N
\end{equation}
the corresponding $1$-particle density. According to Laughlin's plasma analogy~\cite{Laughlin-83} the density profile is for large $N$ well approximated by a droplet of radius $(\ell N)^{1/2}$ and fixed density $(\pi\ell)^{-1}$,
\beq
\varrho^{\rm flat}_N(\rv) := \begin{cases}\frac{1}{\pi \ell} \mbox{ if } |z|\leq \sqrt {\ell N}\\
                               0 \mbox{ otherwise}.
                             \end{cases}
\eeq
Indeed, by a rigorous mean-field analysis   it was proved in \cite{RouSerYng-13b} that this approximation holds in the sense of averages over discs of radius $N^\alpha$ with $1/2>\alpha>1/4$. More generally, the $k$-particle densities are well approximated in this sense by the $k$-fold tensor power of the flat density if $N\to\infty$.
The more refined analysis of classical Coulomb systems in \cite{LebSer-16,BauBouNikYau-15, Ser-20} leads to an extension of this result down to mesoscopic scales $N^\alpha$ for all $\alpha>0$. We shall now use results from \cite{LebSer-16} to estimate the density of Laughlin states in higher Landau levels.

The Laughlin state corresponding to~\eqref{Laugh01} in the $n$-th Landau level $\nLL_N$ is
\begin{equation}
\Psi_{n,N}^{\rm Lau}=c_N\prod_{i<j}(b^\dagger_i-b^\dagger_j)^\nu \varphi_{n,0}^{\otimes N}=
c_N\prod_{i<j}(b^\dagger_i-b^\dagger_j)^\nu \prod_{i=1}^N\left[ (a^\dagger_i)^n\,\varphi_{0,0}\right]
\label{Laughn1}
\end{equation}
with wave function (cf Lemma~\ref{lem:simple})
\begin{equation}
\Psi_{n,N}^{\rm Lau}(\rv_1,\dots,\rv_N) = c_N\left[ \frac 1{(n!)^{N/2}}\prod_{i=1}^N\left(\bar z_i-\partial_{z_i}\right)^n\prod_{i<j}(z_i-z_j)^\nu\right] e^{-(|z_1|^2+\cdots+|z_N|^2)/2}.\label{Laughn2}
\end{equation}
This is, in \emph{electronic position variables}, the exact ground state of a Hamiltonian obtained by
\begin{itemize}
 \item Projecting the physical starting point~\eqref{eq:full hamil} in the $\nLL_N$.
 \item Unitarily mapping the result down to an effective Hamiltonian on $\LLL_N$ using Theorem~\ref{thm:main}.
 \item Neglecting the one-body potential $V_n$ and truncating the Haldane pseudo-potential series of the interaction potential $w_n$. 
\end{itemize}
The Hamiltonian obtained this way acts on $\LLL_N$, and its exact ground state is a LLL Laughlin state in \emph{guiding center variables}. Lifting it back up to the $\nLL_N$ results in~\eqref{Laughn2}:  
\beq\Psi_{n,N}^{\rm Lau} = \left( U_n^* \right) ^{\otimes N}  \Psi_{0,N}^{\rm Lau}.\eeq
We now vindicate a natural expectation: the density of $\Psi_{n,N} ^{\rm Lau}$ is very close, for large $N$, to that of $\Psi_{0,N}^{\rm Lau}$ on length scales much larger than the magnetic length ($1$ in our units). This is because electron coordinates and guiding center coordinates differ only on the scale of a cyclotron orbit, which is much smaller than the thermodynamically large extent of the states themselves.  

We shall test the densities with regularized characteristic functions of discs. Let $\chi_1$ be the characteristic function of the unit disc around the origin and for $\eps>0$ let $\eta_\eps$ be a function with support in the annulus $1\leq |\mathbf r|\leq 1+\eps$ such that
 $\chi_{1,\eps}:=\chi_{1}+\eta_\eps$ is $C^\infty$.
 For $R>0$ and $\mathbf r_0\in \mathbb R^2$ define 
 \beq\chi_{R,{\mathbf r}_0,\eps}(\mathbf r)=\chi_{1,\eps}(R^{-1}(\mathbf r-\mathbf r_0))\label{chiR}.\eeq
The analysis in \cite{RouSerYng-13b,LebSer-16} is carried out using scaled variables, 
\beq\mathbf r'=N^{-1/2}\mathbf r.\eeq
In these variables the extension of the Laughlin state  is $O(1)$ and mesoscopic scales are $O(N^{-\gamma})$ with $0<\gamma<1/2$.  The scaled densities are ($\varrho_{n,N}^{\rm Lau}$ is defined in analogy with~\eqref{eq:Laugh dens})
\begin{equation}
 \widetilde{\varrho}_{n,N}^{\rm Lau}(\mathbf r') = \varrho_{n,N}^{\rm Lau}(N^{1/2}\mathbf r') \mbox{ and }  \widetilde{\varrho}^{\rm flat}_N(\mathbf r')=\varrho^{\rm flat}_N(N^{1/2}\mathbf r'). 
\end{equation}
We scale the test functions accordingly and consider $\chi_{r,{\mathbf r'_0},\eps}(\mathbf r')=\chi_{1,\eps}(r^{-1}(\mathbf r'-\mathbf r'_0))$. The result on the density and its fluctuations we want to sketch the proof of is as follows:

\begin{theorem}[\textbf{Density of Laughlin states on mesoscopic scales}]\label{thm:density}\mbox{}\\
\emph{\textbf{(i)}} For every Landau index $n$, every fixed $\eps>0$  and all mesoscopic scales $r\sim N^{-\gamma}$ with $0<\gamma<\half$
\begin{equation}
 \int \widetilde{\varrho}^{\rm Lau}_{N,n}(\mathbf r') \chi_{r,\eps,z_0}(\mathbf r')\,{\mathrm d}^2\mathbf r'=\int\widetilde{\varrho}^{\rm flat}_N(\mathbf r') \chi_{r,\eps,z_0}(\mathbf r')\, {\mathrm d}^2\mathbf r' \,  (1+O(N^{-1+2\gamma }))\label{128} 
\end{equation}

\noindent\emph{\textbf{(ii)}} If $r\sim N^{-\gamma}$,  the fluctuation of the linear statistics associated to $\chi_{r,\eps,z_0}$  in the $n$-th Landau level is 
\beq \sim \eps^{-1/2}(1+\eps^{-2n} O(N^{-n(1-2\gamma)})).\label{129}\eeq 
\end{theorem}

\begin{proof} 
The considerations of Sec.~\ref{densities} imply that for every test function $\chi$
\beq \int_{\R^2} \widetilde{\varrho}^{\rm Lau}_{N,n}(\mathbf r') \chi (\rv') = \int_{\R^2} \widetilde{\varrho}^{\rm Lau}_{N,0}(\mathbf r') L_n \left(-\mbox{$\frac 14$} N^{-1} \Delta\right)  \chi  (\rv')\eeq
with $L_n$ the Laguerre polynomial. The point is that for large $N$, only the lowest order term in the above polynomial will contribute: 
\beq L_n \left(-\mbox{$\frac 14$} N^{-1} \Delta\right) \approx 1.\eeq
We use Theorem 1 and Remark 1.2 in \cite{LebSer-16}, see also Theorem 1 in \cite{Ser-20}.  The function  denoted there by $\xi$ is in the present case 
\beq
L_n \left(-\mbox{$\frac 14$} N^{-1} \Delta\right) \chi_{r,\eps,\mathbf r_0'}.
\eeq
The potential in Equation (1.14) of~\cite{LebSer-16} is here $|z|^2$ so Mean$(\xi)=0$.  Equation (1.17) in  \cite{LebSer-16} and 
\beq
\int\widetilde\varrho^{\rm flat}_N \chi_{r,\eps,z_0}\sim r^2
\eeq 
now lead directly to \eqref{128} above. The dependence of the error term on $n$ and $\eps$ cannot be deduced from Equation (1.17) in Remark 1.2 alone, however.

For the fluctuations we need $\Vert \nabla \xi\Vert_2$, according to the \lq\lq Mesoscopic case" of Theorem~1 in  \cite{LebSer-16}. Since 
\beq
\xi= L_n \left(-\mbox{$\frac 14$} N^{-1} \Delta\right) \chi_{r,\eps,\mathbf r_0'},
\eeq 
Equation~\eqref{129} is a consequence of this theorem and of a simple $L^2$-estimate of the gradient of the test function $\chi_{r,\eps,\mathbf r_0'}$.
\end{proof}

\subsection{Rigidity estimates}

In~\cite{RouYng-14,RouYng-15,LieRouYng-16,LieRouYng-17,RouYng-17,OlgRou-19} we have investigated rigidity/stability properties of the LLL Laughlin state. The question is now the response of the Laughlin function to a slight relaxation of the assumptions made in it's derivation, namely that one could in first approximation neglect the one-body potential and truncate the Haldane pseudo-potential series to a finite order. If one assumes the validity of a certain ``spectral gap conjecture'' (see~\cite[Appendix]{Rougerie-xedp19} and references therein), investigating this question basically means minimizing the one-body energy and the residual part of the interaction \emph{within the full ground eigenspace} of the truncated interaction energy (cf degenerate perturbation theory). Our main conclusion was that this problem could be solved  to leading order in the  large $N$ limit by generating quasi-holes on top of Laughlin's wave function. We now want to quickly explain how this can be generalized to Laughlin states in higher levels. We discuss only the adaptation of~\cite{LieRouYng-17,RouYng-17} for the response to one-body potential. One could consider as well the response to smooth long-range weak interactions as in~\cite{OlgRou-19}, but for brevity we do not write this explicitly. 
 
We take $v:\R^2 \rightarrow \R^+$ to be a smooth one-body potential, growing polynomially at infinity. We scale it so that it lives on the scale of the Laughlin wave-function: 
\beq V_N (\rv) = v (\sqrt{N} \rv).\eeq
As discussed in the aforementioned references these assumptions can be relaxed to some extent. The main observation is that after the reduction of the $\nLL$ interacting Hamiltonian discussed in Subsection VI A, any multiplication of the $\LLL$ Laughlin state by a symmetric analytic function $F$ still yields an exact zero-energy eigenstate in guiding center variables. It is thus relevant to consider the action of the one-body potential $V_N$ on the ground-state space of the truncated interaction Hamiltonian. In electron variables the latter is 
\begin{equation}\label{eq:GS space}
\cL_{N,n} ^\ell := \left\{ \Psi_{N,n} \in \nLL_N, U_n ^{\otimes N} \Psi_{N,n} = F(z_1,\ldots,z_N) \PsiLau ^{(\ell)} \mbox{ with } F \mbox{ analytic and symmetric }\right\} 
\end{equation}
where the $\LLL$ Laughlin state is as in~\eqref{eq:Laughlin}. For any many-body wave-function $\Psi_{N,n}\in \cL_{N,n}^{\ell}$ we define it's one-particle density as 
\beq \varrho_{\Psi_{N,n}} (\rv) := N \int_{\R^{2(N-1)}} |\Psi_{N,n} (\rv,\rv_2,\ldots,\rv_N)|^2 d\rv_2\ldots d\rv_N.\eeq
The variational problem for the response of the Laughlin state to an external potential, within the class~\eqref{eq:GS space} is now 
\begin{equation}\label{eq:full var prob}
E (N,n,\ell) := \inf \left\{ \int_{R^2} V \varrho_{\Psi_{N,n}}, \, \Psi_{N,n} \in \cL_{N,n} ^{\ell}, \int_{\R^{2N}} |\Psi_{N,n}|^2 = 1  \right\}. 
\end{equation}
It is of importance in Laughlin's theory of the FQHE that one needs only consider so-called quasi-holes states to solve the  above approximately. If one makes this approximation, the minimum energy becomes 
\begin{equation}\label{eq:red var prob}
e (N,n,\ell) := \inf \left\{ \int_{R^2} V \varrho_{\Psi_{N,n}},\, U_n ^{\otimes N} \Psi_{N,n} = f^{\otimes N} \PsiLau ^{(\ell)} \mbox{ with } f \mbox{ analytic }, \int_{\R^{2N}} |\Psi_{N,n}|^2 = 1  \right\}. 
\end{equation}
The latter energy is obtained by reducing the variational set, so, obviously 
\beq E(N,n,\ell) \leq e (N,n,\ell).\eeq
What is much less obvious is that this upper bound is optimal in the large $N$ limit: 

\begin{theorem}[\textbf{Response of higher LL Laughlin states to external potentials}]\label{thm:rigidity}\mbox{}\\
With the previous notation we have, for any fixed $n,\ell \in \N$
\begin{equation}\label{eq:rigidity}
\frac{E(N,n,\ell)}{e (N,n,\ell)} \underset{N \to \infty}{\rightarrow} 1 
\end{equation}
\end{theorem}

The $n=0$ version of the above was proved in~\cite{LieRouYng-17,RouYng-17}. The adaptation to higher $n$ follows from the tools therein, together with the representation of $U_n V_n U_n^* $ discussed at length in Sec.~\ref{sec:proof thm}. We do not give details for brevity. We however point out that consequences for minimizing densities also follow, so that the density of a (quasi)-minimizer for~\eqref{eq:full var prob} is approximately flat with value $(\pi \ell) ^{-1}$ on an open set to be optimized over, and quickly drops to $0$ outside. This is in accordance with the physical picture of the system responding to external potentials by generating quasi-holes to accommodate their crests. Indeed, the interpretation of the states in~\eqref{eq:red var prob} is that the zeroes of the analytic function $f$ correspond to the location of quasi-holes in guiding center coordinates.
 

  \bigskip 
  
\noindent\textbf{Acknowledgements.} We had helpful conversations regarding the material of this paper with Thierry Champel, S\o{}ren Fournais and Alessandro Olgiati. We received funding from the European Research Council (ERC) under the European Union's Horizon 2020 Research and Innovation Programme (Grant agreement CORFRONMAT No 758620).

\bibliographystyle{siam}

\end{document}